\documentclass[proceedings]{stacs}
\stacsheading{2009}{601--612}{Freiburg}
\firstpageno{601}

\newcommand{\MNP}{\text{\textsc{MAX NP}}}
\newcommand{\MSNP}{\text{\textsc{MAX SNP}}}
\newcommand{\MINF}{\text{\textsc{MIN F}}^+\Pi}
\newcommand{\A}{\mathcal{A}}
\newcommand{\F}{\mathcal{F}}
\newcommand{\Q}{\mathcal{Q}}
\renewcommand{\S}{\mathcal{S}}
\newcommand{\h}{\mathcal{H}}

\newcommand{\x}{\mathbf{x}}
\newcommand{\y}{\mathbf{y}}
\newcommand{\z}{\mathbf{z}}
\newcommand{\pQ}{p\text{-}\mathcal{Q}}

\newcommand{\N}{\mathbb{N}}

\DeclareMathOperator{\opt}{opt}

\newcommand{\hide}[1]{}

\begin{document}

\title[Polynomial Kernelizations]{Polynomial Kernelizations for $\MINF_1$ and $\MNP$}

\author[lab1]{S. Kratsch}{Stefan Kratsch}
\address[lab1]{Max-Planck-Institut f\"ur Informatik, Campus E 1 4, 66123 Saarbr\"ucken, Germany}  
\email{skratsch@mpi-inf.mpg.de}  
\urladdr{http://www.mpi-inf.mpg.de/\textasciitilde skratsch}  

\keywords{parameterized complexity, kernelization, approximation algorithms}
\subjclass{F.2.2}

\begin{abstract}
\noindent The relation of constant-factor approximability to fixed-parameter tractability and kernelization is a long-standing open question. We prove that two large classes of constant-factor approximable problems, namely~$\MINF_1$ and~$\MNP$, including the well-known subclass~$\MSNP$, admit polynomial kernelizations for their natural decision versions. This extends results of Cai and Chen (JCSS 1997), stating that the standard parameterizations of problems in~$\MSNP$ and~$\MINF_1$ are fixed-parameter tractable, and complements recent research on problems that do not admit polynomial kernelizations (Bodlaender et al.\ ICALP 2008).
\end{abstract}

\maketitle

\section{Introduction}
The class~APX consists of all~NP optimization problems that are approximable to within a constant factor of the optimum. It is known that the decision versions of most~APX-problems are fixed-parameter tractable or even admit efficient preprocessing in the form of a polynomial kernelization. How strong is the relation between constant-factor approximability and polynomial kernelizability? Is there a property inherent to most~APX-problems that explains this relation? What is the nature of~APX-problems that do not admit a polynomial kernelization, such as \textsc{Bin Packing} for example?

Since many prominent~APX-problems are complete under approximation preserving reductions and do not admit arbitrarily small approximation ratios, studying their parameterized complexity is a natural approach to obtain better results (recently Cai and Huang presented fixed-parameter approximation schemes for~$\MSNP$~\cite{Cai06}). In conjunction with recent work on problems without polynomial kernelizations, positive answers to the questions may provide evidence against~APX-membership for some problems (e.g.\ \textsc{Treewidth}).

\vskip5pt
\noindent\textbf{Our work:} We prove that the standard parameterizations of problems in two large classes of constant-factor approximable problems, namely~$\MINF_1$ and~$\MNP$, admit polynomial kernelizations. This extends results of Cai and Chen~\cite{Cai1997} who showed that the standard parameterizations of all problems in~$\MINF_1$ and~$\MSNP$ (a subclass of~$\MNP$) are fixed-parameter tractable.\footnote[1]{The existence of a kernelization, not necessarily polynomial, is equivalent to fixed-parameter tractability.} Interestingly perhaps, both our results rely on the Sunflower Lemma due to Erd\H os and Rado~\cite{Erdos1960}.
\newpage
\begin{table}[t]
\begin{tabular}{|l|c|c|}
\hline
 & Approximation ratio & Kernel size\\
\hline
\textsc{Minimum Vertex Cover} & $2$ \cite{Halperin2000} & $O(k)$ \cite{Chen1999}\\
\textsc{Feedback Vertex Set} & $2$ \cite{Bafna1999}& $O(k^3)$ \cite{Bodlaender2007}\\
\textsc{Minimum Fill-In} & $O(\opt)$ \cite{Natanzon2000} & $O(k^2)$ \cite{Natanzon2000}\\
\textsc{Treewidth} & $O(\sqrt{\log \opt})$ \cite{Feige2008} & not poly\footnotemark[2] \cite{Bodlaender2008}\\
\hline
\end{tabular}
\caption{\label{table:examples}Approximation ratio and size of problem kernels for some optimization problems.}
\end{table}

\vskip5pt
\noindent\textbf{Related work:} Recently Bodlaender et al.~\cite{Bodlaender2008} presented the first negative results concerning the existence of polynomial kernelizations for some natural fixed-parameter tractable problems. Using the notion of a \emph{distillation algorithm} and results due to Fortnow and Santhanam~\cite{Fortnow2008}, they were able to show that the existence of polynomial kernelizations for so-called \emph{compositional} parameterized problems implies a collapse of the polynomial hierarchy to the third level. These are seminal results presenting the first super-linear lower bounds for kernelization and relating a statement from parameterized complexity to a hypothesis from classical complexity theory.

In Table~$\ref{table:examples}$ we summarize approximability and kernelization results for some well-known problems.

\vskip5pt
\noindent\textbf{$\MINF_1$ and~$\MNP$:} Two decades ago Papadimitriou and Yannakakis~\cite{Papadimitriou1991} initiated the syntactic study of optimization problems to extend the understanding of approximability. They introduced the classes~$\MNP$ and~$\MSNP$ as natural variants of~NP based on Fagin's~\cite{Fagin1974} syntactic characterization of~NP. Essentially problems are in~$\MNP$ or~$\MSNP$ if their optimum value can be expressed as the maximum number of tuples for which some existential, respectively quantifier-free, first-order formula holds. They showed that every problem in these two classes is approximable to within a constant factor of the optimum. Arora et al.\ complemented this by proving that no~$\MSNP$-complete problem has a polynomial-time approximation scheme, unless~P=NP~\cite{Arora1998}. Contained in~$\MSNP$ there are some well-known maximization problems, such as \textsc{Max Cut}, \textsc{Max~$q$-Sat}, and \textsc{Independent Set} on graphs of bounded degree. Its superclass~$\MNP$ also contains \textsc{Max Sat} amongst others.

Kolaitis and Thakur generalized the approach of examining the logical definability of optimization problems and defined further classes of minimization and maximization problems~\cite{Kolaitis1994,Kolaitis1995}. Amongst others they introduced the class~$\MINF_1$ of problems whose optimum can be expressed as the minimum weight of an assignment (i.e.\ number of ones) that satisfies a certain universal first-order formula. They proved that every problem in~$\MINF_1$ is approximable to within a constant factor of the optimum. In~$\MINF_1$ there are problems like \textsc{Vertex Cover},~$d$-\textsc{Hitting Set}, and \textsc{Triangle Edge Deletion}.

\footnotetext[2]{Treewidth does not admit a polynomial kernelization unless there is a distillation algorithm for all coNP complete problems~\cite{Bodlaender2008}. Though unlikely, this is not known to imply a collapse of the polynomial hierarchy.}

\vskip5pt
Section~$\ref{preliminaries}$ covers the definitions of the classes~$\MINF_1$ and~$\MNP$, as well as the necessary details from parameterized complexity. In Sections~$\ref{section:minf}$ and~$\ref{section:maxnp}$ we present polynomial kernelizations for the standard parameterizations of problems in~$\MINF_1$ and~$\MNP$ respectively. Section~$\ref{conclusion}$ summarizes our results and poses some open problems.

\section{Preliminaries}\label{preliminaries}

\noindent\textbf{Logic and complexity classes:}  A (relational) vocabulary is a set~$\sigma$ of relation symbols, each having some fixed integer as its arity. 
Atomic formulas over~$\sigma$ are of the form~$R(z_1,\dots,z_t)$ where~$R$ is a~$t$-ary relation symbol from~$\sigma$ and the~$z_i$ are variables. The set of quantifier-free (relational) formulas over~$\sigma$ is the closure of the set of all atomic formulas under negation, conjunction, and disjunction.

\begin{defi}[$\MINF_1$,~$\MNP$]\label{def:classes}
A \emph{finite structure of type~$(r_1,\dots,r_t)$} is a tuple~$\A=(A,R_1,\dots,R_t)$ where~$A$ is a finite set and each~$R_i$ is an~$r_i$-ary relation over~$A$.

Let~$\Q$ be an optimization problem on finite structures of type~$(r_1,\dots,r_t)$. Let~$R_1,\dots,R_t$ be relation symbols of arity~$r_1,\dots,r_t$.

\noindent(a) The problem~$\Q$ is contained in the class~$\MINF_1$ if its optimum on finite structures~$\A$ of type~$(r_1,\dots,r_t)$ can be expressed as
$$\opt_\Q(\A)=\min_S\lbrace|S|:(\A,S)\models(\forall\x\in A^{c_x}):\psi(\x,S)\rbrace,$$
where~$S$ is a single relation symbol and~$\psi(\x,S)$ is a quantifier-free formula in conjunctive normal form over the vocabulary~$\lbrace R_1,\dots,R_t,S\rbrace$ on variables~$\lbrace x_1,\dots,x_{c_x}\rbrace$. Furthermore,~$\psi(\x,S)$ is positive in~$S$, i.e.~$S$ does not occur negated in~$\psi(\x,S)$.

\noindent(b) The problem~$\Q$ is contained in the class~$\MNP$ if its optimum on finite structures~$\A$ of type~$(r_1,\dots,r_t)$ can be expressed as
$$\opt_\Q(\A)=\max_\S\left\vert\left\lbrace\x\in A^{c_x}:(\A,\S)\models(\exists\y\in A^{c_y}):\psi(\x,\y,\S)\right\rbrace\right\vert,$$
where~$\S=(S_1,\dots,S_u)$ is a tuple of~$s_i$-ary relation symbols~$S_i$ and~$\psi(\x,\y,\S)$ is a quantifier-free formula in disjunctive normal form over the vocabulary~$\lbrace R_1,\dots,R_t,S_1,\dots,S_u\rbrace$ on variables~$\lbrace x_1,\dots,x_{c_x},y_1,\dots,y_{c_y}\rbrace$.
\end{defi}

\begin{remark}
The definition of~$\MSNP$ is similar to that of~$\MNP$ but without the existential quantification of~$\y$, i.e~$\opt_\Q(\A)=\max_\S|\lbrace\x:(\A,\S)\models\psi(\x,\S)\rbrace|$.
\end{remark}

\begin{example}[\textsc{Minimum Vertex Cover}] Let~$G=(V,E)$ be a finite structure of type~$(2)$ that represents a graph by a set~$V$ of vertices and a binary relation~$E$ over~$V$ as its edges. The optimum of \textsc{Minimum Vertex Cover} on structures~$G$ can be expressed as:
$$\opt_{VC}(G)=\min_{S\subseteq V}\lbrace|S|:(G,S)\models(\forall(u,v)\in V^2):(\neg E(u,v)\vee S(u)\vee S(v))\rbrace.$$
This implies that \textsc{Minimum Vertex Cover} is contained in~$\MINF_1$.
\end{example}

\begin{example}[\textsc{Maximum Satisfiability}] Formulas in conjunctive normal form can be represented by finite structures~$\F=(F,P,N)$ of type~$(2,2)$: Let~$F$ be the set of all clauses and variables, and let~$P$ and~$N$ be binary relations over~$F$. Let~$P(x,c)$ be true if and only if~$x$ is a literal of the clause~$c$ and let~$N(x,c)$ be true if and only if~$\neg x$ is a literal of the clause~$c$. The optimum of \textsc{Max Sat} on structures~$\F$ can be expressed as:  
$$\opt_{MS}(\F)=\\\max_{T\subseteq F}|\lbrace c\in F:(\F,T)\models(\exists x\in F):(P(x,c)\wedge T(x))\vee(N(x,c)\wedge\neg T(x))\rbrace|.$$
Thus \textsc{Max Sat} is contained in~$\MNP$.
\end{example}

For a detailed introduction to~$\MINF_1$,~$\MNP$, and~$\MSNP$ we refer the reader to~\cite{Kolaitis1994,Kolaitis1995,Papadimitriou1991}. An introduction to logic and complexity can be found in~\cite{Papadimitriou1993}.

\vskip4pt
\noindent\textbf{Parameterized complexity:}  The field of parameterized complexity, pioneered by Downey and Fellows, is a two-dimensional approach of coping with combinatorially hard problems. Parameterized problems come with a parameterization that maps input instances to a parameter value. The time complexity of algorithms is measured with respect to the input size and the parameter. In the following we give the necessary formal definitions, namely fixed-parameter tractability, standard parameterizations, and kernelization.

\begin{defi}[Fixed-parameter tractability]
A \emph{parameterization} of~$\Sigma^*$ is a polynomial-time computable mapping~$\kappa:\Sigma^*\to\N$. A \emph{parameterized problem} over an alphabet~$\Sigma$ is a pair~$(\Q,\kappa)$ consisting of a set~$\Q\subseteq\Sigma^*$ and a parameterization~$\kappa$ of~$\Sigma^*$.

A parameterized problem~$(\Q,\kappa)$ is \emph{fixed-parameter tractable} if there exists an algorithm~$\mathbb{A}$, a polynomial~$p$, and a computable function~$f:\N\to\N$ such that~$\mathbb{A}$ decides~$x\in\Q$ in time~$f(\kappa(x))\cdot p(|x|)$. FPT is the class of all fixed-parameter tractable problems.
\end{defi}

\begin{defi}[Standard parameterization]
Let~$\Q$ be a maximization (minimization) problem. The \emph{standard parameterization} of~$\Q$ is~$\pQ=(d\text{-}\mathcal{Q},\kappa)$ where~$\kappa:(\A,k)\mapsto k$ and~$d\text{-}\mathcal{Q}$ is the language of all tuples~$(\A,k)$ such that~$\opt_\Q(\A)\geq k$ ($\opt_\Q(\A)\leq k$).
\end{defi}

Basically~$d\text{-}\mathcal{Q}$ is the decision version of~$\Q$, asking whether the optimum is at least~$k$ (respectively at most~$k$). The standard parameterization of~$\Q$ is~$d\text{-}\mathcal{Q}$ parameterized by~$k$.

\begin{defi}[Kernelization]
Let~$(\Q,\kappa)$ be a parameterized problem over~$\Sigma$. A poly\-nomi\-al-time computable function~$K:\Sigma^*\to\Sigma^*$ is a \emph{kernelization} of~$(\Q,\kappa)$ if there is a computable function~$h:\N\to\N$ such that for all~$x\in\Sigma^*$ we have
$$(x\in\Q\Leftrightarrow K(x)\in\Q)\text{ and }|K(x)|\leq h(\kappa(x)).$$
We call~$h$ the \emph{size} of the \emph{problem kernel}~$K(x)$. The kernelization~$K$ is \emph{polynomial} if~$h$ is a polynomial. We say that~$(\Q,\kappa)$ \emph{admits a (polynomial) kernelization} if there exists a (polynomial) kernelization of~$(\Q,\kappa)$.
\end{defi}

Essentially, a kernelization is a polynomial-time data reduction that comes with a guaranteed upper bound on the size of the resulting instance in terms of the parameter.

For an introduction to parameterized complexity we refer the reader to~\cite{DF2006,Flum2006,Niedermeier2006}.

\vskip4pt
\noindent\textbf{Hypergraphs and sunflowers:} We assume the reader to be familiar with the basic graph notation. A \emph{hypergraph} is a tuple~$\h=(V,E)$ consisting of a finite set~$V$, its vertices, and a family~$E$ of subsets of~$V$, its edges. A hypergraph has \emph{dimension}~$d$ if each edge has cardinality at most~$d$. A hypergraph is~$d$-\emph{uniform} if each edge has cardinality exactly~$d$.

\begin{defi}[Sunflower]
Let~$\h$ be a hypergraph. A \emph{sunflower} of cardinality~$r$ is a set~$F=\lbrace f_1,\dots,f_r\rbrace$ of edges of~$\h$ such that every pair has the same intersection~$C$, i.e.\ for all~$1\leq i<j\leq r$: $f_i\cap f_j=C$.
The set~$C$ is called the \emph{core} of the sunflower.
\end{defi}

Note that any family of pairwise disjoint sets is a sunflower with core~$C=\emptyset$.

\begin{lem}[Sunflower Lemma~\cite{Erdos1960}]
Let~$k,d\in\N$ and let~$\h$ be a~$d$-uniform hypergraph with more than~$(k-1)^d\cdot d!$ edges. Then there is a sunflower of cardinality~$k$ in~$\h$. For every fixed~$d$ there is an algorithm that computes such a sunflower in time polynomial in~$|E(\h)|$.
\end{lem}

\begin{cor}[Sunflower Corollary]
The same holds for~$d$-dimensional hypergraphs with more than~$(k-1)^d\cdot d!\cdot d$ edges.
\end{cor}

\begin{proof}
For some~$d'\in\lbrace1,\dots,d\rbrace$,~$\h$ has more than~$(k-1)^d\cdot d!\geq(k-1)^{d'}\cdot d'!$ edges of cardinality~$d'$. Let~$\h_{d'}$ be the~$d'$-uniform subgraph induced by the edges of cardinality~$d'$. We apply the Sunflower Lemma on~$\h_{d'}$ and obtain a sunflower~$F$ of cardinality~$k$ in time polynomial in~$|E(\h_{d'})|\leq |E(\h)|$. Clearly~$F$ is also a sunflower of~$\h$.
\end{proof}

\section{Polynomial kernelization for $\MINF_1$}\label{section:minf}

We will prove that the standard parameterization of any problem in~$\MINF_1$ admits a polynomial kernelization. The class~$\MINF_1$ was introduced by Kolaitis and Thakur in a framework of syntactically defined classes of optimization problems~\cite{Kolaitis1994}. In a follow-up paper they showed that every problem in~$\MINF_1$ is constant-factor approximable~\cite{Kolaitis1995}.

Throughout the section let~$\Q\in\MINF_1$ be an optimization problem on finite structures of type~$(r_1,\dots,r_t)$. Let~$R_1,\dots,R_t$ be relation symbols of arity~$r_1,\dots,r_t$ and let~$S$ be a relation symbol of arity~$c_S$. Furthermore, let~$\psi(\x,S)$ be a quantifier-free formula in conjunctive normal form over the vocabulary~$\lbrace R_1,\dots,R_t,S\rbrace$ on variables~$\lbrace x_1,\dots,x_{c_x}\rbrace$ that is positive in~$S$ such that
$$\opt_\Q(\A)=\min_{S\subseteq A^{c_S}}\lbrace\vert S\vert:(\A,S)\models(\forall\x\in A^{c_x}):\psi(\x,S)\rbrace.$$

Let~$s$ be the maximum number of occurrences of~$S$ in any clause of~$\psi(\x,S)$. The standard parameterization~$\pQ$ of~$\Q$ is the following problem:
\begin{flushleft}\begin{tabular}{ll}
	\textbf{Input:}&A finite structure~$\A$ of type~$(r_1,\dots,r_t)$ and an integer~$k$.\\
	\textbf{Parameter:}&$k$.\\
	\textbf{Task:}& Decide whether~$\opt_\Q(\A)\leq k$.
\end{tabular}\end{flushleft}

We will see that, given an instance~$(\A,k)$, deciding whether~$\opt_\Q(\A)\leq k$ is equivalent to deciding an instance of~$s$-\textsc{Hitting Set}.\footnote[3]{In literature the problem is often called~$d$-\textsc{Hitting Set} but we will need~$d=s$.}  Our kernelization will therefore make use of existing kernelization results for~$s$-\textsc{Hitting Set}. The parameterized version of~$s$-\textsc{Hitting Set} is defined as follows:

\begin{flushleft}\begin{tabular}{ll}
	\textbf{Input:}&A hypergraph~$\h=(V,E)$ of dimension~$s$ and an integer~$k$.\\
	\textbf{Parameter:}&$k$.\\
	\textbf{Task:}& Decide whether~$\h$ has a hitting set of size at most~$k$, i.e.~$S\subseteq V$,~$|S|\leq k$,\\& such that~$S$ has a nonempty intersection with every edge of~$\h$.
\end{tabular}\end{flushleft}

We consider the formula~$\psi(\x,S)$ and a fixed instance~$(\A,k)$, with~$\A=(A,R_1,\dots,R_t)$. For every tuple~$\x\in A^{c_x}$ we can evaluate all literals of the form~$R_i(\z)$ and~$\neg R_i(\z)$ for some~$\z\in\lbrace x_1,\dots,x_{c_x}\rbrace^{r_i}$. By checking whether~$\z\in R_i$, we obtain~$1$ (true) or~$0$ (false) for each literal. Then we delete all occurrences of~$0$ from the clauses and delete all clauses that contain a~$1$. For each~$\x$, we obtain an equivalent formula that we denote with~$\psi_\x(S)$. Each~$\psi_\x(S)$ is in conjunctive normal form on literals~$S(\z)$ for some~$\z\in\lbrace x_1,\dots,x_{c_x}\rbrace^{c_S}$ (no literals of the form~$\neg S(\z)$ since~$\psi(\x,S)$ is positive in~$S$).

\begin{remark}\label{remark:psix}
For all~$\x\in A^{c_x}$ and~$S\subseteq A^{c_S}$ it holds that~$(\A,S)\models\psi(\x,S)$ if and only if~$(\A,S)\models\psi_\x(S)$. Moreover, we can compute all formulas~$\psi_\x(S)$ for~$\x\in A^{c_x}$ in polynomial time, since~$c_x$ and the length of~$\psi(\x,S)$ are constants independent of~$\A$.
\end{remark}

Deriving a formula~$\psi_\x(S)$ can yield empty clauses. This happens when all literals~$R_i(\cdot)$,~$\neg R_i(\cdot)$ in a clause are evaluated to~$0$ and there are no literals~$S(\cdot)$. In that case, no assignment~$S$ can satisfy the formula~$\psi_\x(S)$, or equivalently~$\psi(\x,S)$. Thus~$(\A,k)$ is a no-instance. Note that clauses of~$\psi_\x(S)$ cannot contain contradicting literals since~$\psi(\x,S)$ is positive in~$S$.

\begin{remark}
From now on, we assume that all clauses of the formulas~$\psi_\x(S)$ are nonempty.
\end{remark}

We define a mapping~$\Phi$ from finite structures~$\A$ to hypergraphs~$\h$. Then we show that equivalent~$s$-\textsc{Hitting Set} instances can be obtained in this way. 
\begin{defi}\label{def:eqhs}
Let~$\A$ be an instance of~$\Q$. We define~$\Phi(\A):=\h$ with~$\h=(V,E)$. We let~$E$ be the family of all sets~$e=\lbrace \z_1,\dots,\z_p\rbrace$ such that~$(S(\z_1)\vee\dots\vee S(\z_p))$ is a clause of a~$\psi_\x(S)$ for some~$\x\in A^{c_x}$. We let~$V$ be the union of all sets~$e\in E$.
\end{defi}

\begin{remark}
The hypergraphs~$\h$ obtained from the mapping~$\Phi$ have dimension~$s$ since each~$\psi_\x(S)$ has at most~$s$ literals per clause. It follows from Remark~$\ref{remark:psix}$ that~$\Phi(\A)$ can be computed in polynomial time.
\end{remark}

The following lemma establishes that~$(\A,k)$ and~$(\h,k)=(\Phi(\A),k)$ are equivalent in the sense that~$(\A,k)\in\pQ$ if and only if~$(\h,k)\in s$-\textsc{Hitting Set}.

\begin{lem}\label{lemma:eqhs}Let~$\A=(A,R_1,\dots,R_t)$ be an instance of~$\Q$ then for all~$S\subseteq A^{c_S}$:
$$(\A,S)\models(\forall\x):\psi(\x,S)\text{ if and only if }S\text{ is a hitting set for }\h=\Phi(\A).$$
\end{lem}
\begin{proof}
Let~$\h=\Phi(\A)=(V,E)$ and let~$S\subseteq A^{c_S}$:
$$\begin{array}{rl}
	&(\A,S)\models(\forall\x\in A^{c_x}):\psi(\x,S)\\
	\Leftrightarrow&(\A,S)\models(\forall\x\in A^{c_x}):\psi_\x(S)\\
	\Leftrightarrow&(\forall\x\in A^{c_x}):\text{each clause of }\psi_\x(S)\text{ has a literal }S(\z)\text{ for which }\z\in S\\
	\Leftrightarrow&S\text{ has a nonempty intersection with every set }e\in E\\
	\Leftrightarrow&S\text{ is a hitting set for }(V,E).\end{array}$$\vspace*{-0.8cm}

\end{proof}

Our kernelization will consist of the following steps:
\begin{enumerate}
\item Map the given instance~$(\A,k)$ for~$\pQ$ to an equivalent instance~$(\h,k)=(\Phi(\A),k)$ for~$s$-\textsc{Hitting Set} according to Definition~$\ref{def:eqhs}$ and Lemma~$\ref{lemma:eqhs}$.
\item Use a polynomial kernelization for~$s$-\textsc{Hitting Set} on~$(\h,k)$ to obtain an equivalent instance~$(\h',k)$ with size polynomial in~$k$.
\item Use~$(\h',k)$ to derive an equivalent instance~$(\A',k)$ of~$\pQ$. That way we will be able to conclude that~$(\A',k)$ is equivalent to~$(\h,k)$ and hence also to~$(\A,k)$.
\end{enumerate}

There exist different kernelizations for~$s$-\textsc{Hitting Set}: one by Flum and Grohe~\cite{Flum2006} based on the Sunflower Lemma due to Erd\H{o}s and Rado~\cite{Erdos1960}, one by Nishimura et al.~\cite{Nishimura2004} via a generalization of the Nemhauser-Trotter kernelization for \textsc{Vertex Cover}, and a recent one by Abu-Khzam~\cite{Abu-Khzam2007} based on crown decompositions. For our purposes of deriving an equivalent instance for~$\pQ$, these kernelizations have the drawback of shrinking sets during the reduction. This is not possible for our approach since we would need to change the formula~$\psi(\x,S)$ to shrink the clauses. We prefer to modify Flum and Grohe's kernelization such that it uses only edge deletions.

\begin{thm}\label{thm:weakkernel}
There exists a polynomial kernelization of~$s$-\textsc{Hitting Set} that, given an instance~$(\h,k)$, computes an instance~$(\h^*,k)$ such that~$E(\h^*)\subseteq E(\h)$,~$\h^*$ has~$O(k^s)$ edges, and the size of~$(\h^*,k)$ is~$O(k^s)$ as well.
\end{thm}
\begin{proof}
Let~$(\h,k)$ be an instance of~$s$-\textsc{Hitting Set}, with~$\h=(V,E)$. If~$\h$ contains a sunflower~$F=\lbrace f_1,\dots,f_{k+1}\rbrace$ of cardinality~$k+1$ then every hitting set of~$\h$ must have a nonempty intersection with the core~$C$ of~$F$ or with the~$k+1$ disjoint sets~$f_1\setminus C,\dots,f_{k+1}\setminus C$. Thus every hitting set of at most~$k$ elements must have a nonempty intersection with~$C$.

Now consider a sunflower~$F=\lbrace f_1,\dots,f_{k+1},f_{k+2}\rbrace$ of cardinality~$k+2$ in~$\h$ and let~$\h'=(V,E\setminus\lbrace f_{k+2}\rbrace)$. We show that the instances~$(\h,k)$ and~$(\h',k)$ are equivalent. Clearly every hitting set for~$\h$ is also a hitting set for~$\h'$ since~$E(\h')\subseteq E(\h)$. Let~$S\subseteq V$ be a hitting set of size at most~$k$ for~$\h'$. Since~$F\setminus\lbrace f_{k+2}\rbrace$ is a sunflower of cardinality~$k+1$ in~$\h'$, it follows that~$S$ has a nonempty intersection with its core~$C$. Hence~$S$ has a nonempty intersection with~$f_{k+2}\supseteq C$ too. Thus~$S$ is a hitting set of size at most~$k$ for~$\h$, implying that~$(\h,k)$ and~$(\h',k)$ are equivalent.

We start with~$\h^*=\h$ and repeat the following step while~$\h^*$ has more than~$(k+1)^{s}\cdot s!\cdot s$ edges. By the Sunflower Corollary we obtain a sunflower of cardinality~$k+2$ in~$\h^*$ in time polynomial in~$|E(\h^*)|$. We delete an edge of the detected sunflower from the edge set of~$\h^*$ (thereby reducing the cardinality of the sunflower to~$k+1$). Thus, by the argument from the previous paragraph, we maintain that~$(\h,k)$ and~$(\h^*,k)$ are equivalent.

Furthermore~$E(\h^*)\subseteq E(\h)$ and~$\h^*$ has no more than~$(k+1)^s\cdot s!\cdot s\in O(k^s)$ edges. Since we delete an edge of~$\h^*$ in each step, there are~$O(|E(\h)|)$ steps, and the total time is polynomial in~$|E(\h)|$. Deleting all isolated vertices from~$\h^*$ yields a size of~$O(s\cdot k^s)=O(k^s)$ since each edge contains at most~$s$ vertices.
\end{proof}

The following lemma proves that every~$s$-\textsc{Hitting Set} instance that is ``sandwiched'' between two equivalent instances must be equivalent to both.

\begin{lem}\label{lemma:swhs} Let~$(\h,k)$ be an instance of~$s$-\textsc{Hitting Set} and let~$(\h^*,k)$ be an equivalent instance with~$E(\h^*)\subseteq E(\h)$. Then for any~$\h'$ with~$E(\h^*)\subseteq E(\h')\subseteq E(\h)$ the instance~$(\h',k)$ is equivalent to~$(\h,k)$ and~$(\h^*,k)$.
\end{lem}

\begin{proof}
Observe that hitting sets for~$\h$ can be projected to hitting sets for~$\h'$ (i.e.\ restricted to the vertex set of~$\h'$) since~$E(\h')\subseteq E(\h)$. Thus if~$(\h,k)$ is a yes-instance then~$(\h',k)$ is a yes-instance too. The same argument holds for~$(\h',k)$ and~$(\h^*,k)$. Together with the fact that~$(\h,k)$ and~$(\h^*,k)$ are equivalent, this proves the lemma.
\end{proof}

Now we are well equipped to prove that~$\pQ$ admits a polynomial kernelization.

\begin{thm}
Let~$\Q\in\MINF_1$. The standard parameterization~$\pQ$ of~$\Q$ admits a polynomial kernelization.
\end{thm}
\begin{proof}
Let~$(\A,k)$ be an instance of~$\pQ$. By Lemma~$\ref{lemma:eqhs}$ we have that~$(\A,k)$ is a yes-instance of~$\pQ$ if and only if~$(\h,k)=(\Phi(\A),k))$ is a yes-instance of~$s$-\textsc{Hitting Set}. We apply the kernelization from Theorem~$\ref{thm:weakkernel}$ to~$(\h,k)$ and obtain an equivalent~$s$-\textsc{Hitting Set} instance~$(\h^*,k)$ such that~$E(\h^*)\subseteq E(\h)$ and~$\h^*$ has~$O(k^s)$ edges.

Recall that every edge of~$\h$, say~$\lbrace\z_1,\dots,\z_p\rbrace$, corresponds to a clause~$(S(\z_1)\vee\dots\vee S(\z_p))$ of~$\psi_\x(S)$ for some~$\x\in A^{c_x}$. Thus for each edge~$e\in E(\h^*)\subseteq E(\h)$ we can select a tuple~$\x_e$ such that~$e$ corresponds to a clause of~$\psi_{\x_e}(S)$. Let~$X$ be the set of the selected tuples~$\x_e$ for all edges~$e\in E(\h^*)$. Let~$A'\subseteq A$ be the set of all components of tuples~$\x_e\in X$, ensuring that~$X\subseteq A'^{c_x}$. Let~$R'_i$ be the restriction of~$R_i$ to~$A'$ and let~$\A'=(A',R'_1,\dots,R'_t)$.

Let~$(\h',k)=(\Phi(\A'),k)$. By definition of~$\Phi$ and by construction of~$\h'$ we know that~$E(\h^*)\subseteq E(\h')\subseteq E(\h)$ since~$X\subseteq A'^{c_x}$ and~$A'\subseteq A$. Thus, by Lemma~$\ref{lemma:swhs}$, we have that~$(\h',k)$ is equivalent to~$(\h,k)$. Furthermore, by Lemma~$\ref{lemma:eqhs}$,~$(\h',k)$ is a yes-instance of~$s$-\textsc{Hitting Set} if and only if~$(\A',k)$ is a yes-instance of~$\pQ$. Thus~$(\A',k)$ and~$(\A,k)$ are equivalent instances of~$\pQ$.

We conclude the proof by giving an upper bound on the size of~$(\A',k)$ that is polynomial in~$k$. The set~$X$ contains at most~$|E(\h^*)|\in O(k^s)$ tuples. These tuples have no more than~$c_x\cdot|E(\h^*)|$ different components. Hence the size of~$A'$ is~$O(c_x\cdot k^s)=O(k^s)$. Thus the size of~$(\A',k)$ is~$O(k^{sm})$, where~$m$ is the largest arity of a relation~$R_i$. The values~$c_x$,~$s$, and~$m$ are constants that are independent of the input~$(\A,k)$. Thus~$(\A',k)$ is an instance equivalent to~$(\A,k)$ with size polynomial in~$k$.
\end{proof}

\section{Polynomial kernelization for $\MNP$}\label{section:maxnp}

We prove that the standard parameterization of any problem in~$\MNP$ admits a polynomial kernelization. The class~$\MNP$ was introduced by Papadimitriou and Yannakakis in~\cite{Papadimitriou1991}. They showed that every problem in~$\MNP$ is constant-factor approximable.

Throughout the section let~$\Q\in\MNP$ be an optimization problem on finite structures of type~$(r_1,\dots,r_t)$. Let~$R_1,\dots,R_t$ be relation symbols of arity~$r_1,\dots,r_t$ and let~$\S=(S_1,\dots,S_u)$ be a tuple of relation symbols of arity~$s_1,\dots,s_u$. Let~$\psi(\x,\y,\S)$ be a formula in disjunctive normal form over the vocabulary~$\lbrace R_1,\dots,R_t,S_1,\dots,S_u\rbrace$ on variables~$\lbrace x_1,\dots,x_{c_x},y_1,\dots,y_{c_y}\rbrace$ such that for all finite structures~$\A$ of type~$(r_1,\dots,r_t)$:
$$\opt_\Q(\A)=\max_\S|\lbrace\x\in A^{c_x}:(\A,\S)\models(\exists\y\in A^{c_y}):\psi(\x,\y,\S)\rbrace|.$$
Let~$s$ be the maximum number of occurrences of relations~$S_1,\dots,S_u$ in any disjunct of~$\psi(\x,\y,\S)$. The standard parameterization~$\pQ$ of~$\Q$ is the following problem:

\begin{flushleft}
	\begin{tabular}{ll}
	\textbf{Input:}&A finite structure~$\A$ of type~$(r_1,\dots,r_t)$ and an integer~$k$.\\
	\textbf{Parameter:}&$k$.\\
	\textbf{Task:}& Decide whether~$\opt_\Q(\A)\geq k$.
\end{tabular}
\end{flushleft}

Similarly to the previous section, we consider the formula~$\psi(\x,\y,\S)$ and a fixed instance~$(\A,k)$ with~$\A=(A,R_1,\dots,R_t)$. We select tuples~$\x\in A^{c_x}$ and~$\y\in A^{c_y}$ and evaluate all literals of the form~$R_i(\z)$ and~$\neg R_i(\z)$ for some~$\z\in\lbrace x_1,\dots,x_{x_c},y_1,\dots,y_{c_y}\rbrace^{r_i}$. By checking whether~$\z\in R_i$ we obtain~$1$ (true) or~$0$ (false) for each literal. Since~$\psi(\x,\y,\S)$ is in disjunctive normal form, we delete all occurrences of~$1$ from the disjuncts and delete all disjuncts that contain a~$0$. Furthermore, we delete all disjuncts that contain contradicting literals~$S_j(\z),\neg S_j(\z)$ since they cannot be satisfied. We explicitly allow empty disjuncts that are satisfied by definition for the sake of simplicity (they occur when all literals in a disjunct are evaluated to~$1$). We obtain an equivalent formula that we denote with~$\psi_{\x,\y}(\S)$.

\begin{remark}\label{remark:personalize}
For all~$\x$,~$\y$, and~$\S$ it holds that $(\A,\S)\models\psi(\x,\y,\S)$ iff~$(\A,\S)\models\psi_{\x,\y}(\S)$. Moreover, we can compute all formulas~$\psi_{\x,\y}(\S)$ for~$\x\in A^{c_x}$,~$\y\in A^{c_y}$ in polynomial time, since~$c_x$,~$c_y$, and the length of~$\psi(\x,\y,\S)$ are constants independent of~$\A$.
\end{remark}

\begin{defi}Let~$\A=(A,R_1,\dots,R_t)$ be a finite structure of type~$(r_1,\dots,r_t)$.

\noindent(a) We define~$X_\A\subseteq A^{c_x}$ as the set of all tuples~$\x$ such that~$(\exists\y):\psi_{\x,\y}(\S)$ holds for some~$\S$:
$$X_\A=\lbrace\x:(\exists\S):(\A,\S)\models(\exists\y):\psi_{\x,\y}(\S)\rbrace.$$
\noindent(b) For~$\x\in A^{c_x}$ we define~$Y_\A(\x)$ as the set of all tuples~$\y$ such that~$\psi_{\x,\y}(\S)$ holds for some~$\S$:
$$Y_\A(\x)=\lbrace\y:(\exists\S):(\A,\S)\models\psi_{\x,\y}(\S)\rbrace.$$
\end{defi}

\begin{remark}
The sets~$X_\A$ and~$Y_\A(\x)$ can be computed in polynomial time since the number of tuples~$\x\in A^{c_x}$ and~$\y\in A^{c_y}$ is polynomial in the size of~$A$ and~$\psi(\x,\y,\S)$ is of constant length independent of~$\A$.
\end{remark}

\begin{lem}\label{lemma:claim1}
Let~$(\A,k)$ be an instance of~$\pQ$. If~$|X_\A|\geq k\cdot2^s$ then~$\opt_\Q(\A)\geq k$, i.e.~$(\A,k)$ is a yes-instance.
\end{lem}

\begin{proof}
The lemma can be concluded from the proof of the constant-factor approximability of problems in~$\MNP$ in~\cite{Papadimitriou1991}. For each~$\x\in X_\A$ we fix a tuple~$\y\in Y_\A(\x)$ such that~$\psi_{\x,\y}(\S)$ is satisfiable. This yields~$|X_\A|$ formulas, say~$\psi_1,\dots,\psi_{|X_\A|}$. Papadimitriou and Yannakakis showed that one can efficiently compute an assignment that satisfies at least~$\sum{f_i}$ of these formulas, where~$f_i$ is the fraction of all assignments that satisfies~$\psi_i$.

To see that~$f_i\geq 2^{-s}$; consider such a formula~$\psi_i$. Since~$\psi_i$ is satisfiable there exists a satisfiable disjunct. To satisfy a disjunct of at most~$s$ literals, at most~$s$ variables need to be assigned accordingly. Since the assignment to all other variables can be arbitrary this implies that~$f_i\geq2^{-s}$. Thus we have that~$\sum{f_i}\geq|X_\A|\cdot2^{-s}$. Therefore~$|X_\A|\geq k\cdot2^s$ implies that the assignment satisfies at least~$k$ formulas, i.e.\ that~$\opt_\Q(\A)\geq k$.
\end{proof} 

Henceforth we assume that~$|X_\A|<k\cdot2^s$.

\begin{defi}\label{def:sod}
Let~$(\A,k)$ be an instance of~$\pQ$ with~$\A=(A,R_1,\dots,R_t)$. For~$\x\in A^{c_x}$ we define~$D_\A(\x)$ as the set of all disjuncts of~$\psi_{\x,\y}(\S)$ for~$\y\in Y_\A(\x)$.
\end{defi}

\begin{defi}
We define the \emph{intersection of two disjuncts} as the conjunction of all literals that occur in both disjuncts. A \emph{sunflower of a set of disjuncts} is a subset  such that each pair of disjuncts in the subset has the same intersection (modulo permutation of the literals).
\end{defi}

\begin{remark}
The size of each~$D_\A(\x)$ is bounded by the size of~$Y_\A(\x)\subseteq A^{c_y}$ times the number of disjuncts of~$\psi(\x,\y,\S)$ which is a constant independent of~$\A$. Thus the size of each~$D_\A(\x)$ is bounded by a polynomial in the input size. The definition of intersection and sunflowers among disjuncts is a direct analog that treats disjuncts as sets of literals.
\end{remark}

\begin{defi}
A \emph{partial assignment} is a set~$L$ of literals such that no literal is the negation of another literal in~$L$. A formula is \emph{satisfiable under}~$L$ if there exists an assignment that satisfies the formula and each literal in~$L$.
\end{defi}

\begin{prop}\label{prop:dstar}
Let~$(\A,k)$ be an instance of~$\pQ$. For each~$\x\in A^{c_x}$ there exists a set~$D^*_\A(\x)\subseteq D_\A(\x)$ of cardinality~$O(k^s)$ such that:
\begin{enumerate}
 \item For every partial assignment~$L$ of at most~$sk$ literals,~$D^*_\A(\x)$ contains a disjunct satisfiable under~$L$, if and only if~$D_\A(\x)$ contains a disjunct satisfiable under~$L$.
 \item $D^*_\A(\x)$ can be computed in time polynomial in~$|\A|$.
\end{enumerate}
\end{prop}

\begin{proof}
Let~$\A=(A,R_1,\dots,R_t)$ be a finite structure of type~$(r_1,\dots,r_t)$, let~$\x\in A^{c_x}$, and let~$D_\A(\x)$ be a set of disjuncts according to Definition~$\ref{def:sod}$. From the Sunflower Corollary we can derive a polynomial-time algorithm that computes a set~$D^*_\A(\x)$ by successively shrinking sunflowers. We start by setting~$D^*_\A(\x)=D_\A(\x)$ and apply the following step while the cardinality of~$D^*_\A(\x)$ is greater than~$(sk+1)^s\cdot s!\cdot s$.

We compute a sunflower of cardinality~$sk+2$, say~$F=\lbrace f_1,\dots,f_{sk+2}\rbrace$, in time polynomial in~$|D^*_\A(\x)|$ (Sunflower Corollary). We delete a disjunct of~$F$, say~$f_{sk+2}$, from~$D^*_\A(\x)$. Let~$O$ and~$P$ be copies of~$D^*_\A$ before respectively after deleting~$f_{sk+2}$. Observe that~$F'=F\setminus\lbrace f_{sk+2}\rbrace$ is a sunflower of cardinality~$sk+1$ in~$P$. Let~$L$ be a partial assignment of at most~$sk$ literals and assume that no disjunct in~$P$ is satisfiable under~$L$. This means that for each disjunct of~$P$ there is a literal in~$L$ that contradicts it, i.e.\ a literal that is the negation of a literal in the disjunct. We focus on the sunflower~$F'$ in~$P$. There must be a literal in~$L$, say~$l$, that contradicts at least two disjuncts of~$F'$, say~$f$ and~$f'$, since~$|F'|=sk+1$ and~$|L|\leq sk$. Therefore~$l$ is the negation of a literal in the intersection of~$f$ and~$f'$, i.e.\ the core of~$F'$. Thus~$l$ contradicts also~$f_{sk+2}$ and we conclude that no disjunct in~$O=P\cup\lbrace f_{sk+2}\rbrace$ is satisfiable under the partial assignment~$L$. The reverse argument holds since all disjuncts of~$P$ are contained in~$O$. Thus each step maintains the desired property~$(1)$.

At the end~$D^*_\A(x)$ contains no more than~$(sk+1)^s\cdot s!\cdot s\in O(k^s)$ disjuncts. For each~$\x$ this takes time polynomial in the size of the input since the cardinality of~$D_\A(\x)$ is bounded by a polynomial in the input size and a disjunct is deleted in each step.
\end{proof}

\begin{lem}\label{lemma:sandwich} 
Let~$D'_\A(\x)$ be a subset of~$D_\A(\x)$ such that~$D^*_\A(\x)\subseteq D'_\A(\x)\subseteq D_\A(\x)$. For any partial assignment~$L$ of at most~$sk$ literals it holds that~$D_\A(\x)$ contains a disjunct satisfiable under~$L$ if and only if~$D'_\A(\x)$ contains a disjunct satisfiable under~$L$.
\end{lem}

\begin{proof}
Let~$L$ be a partial assignment of at most~$sk$ literals. If~$D_\A(\x)$ contains a disjunct satisfiable under~$L$, then, by Proposition~$\ref{prop:dstar}$, this holds also for~$D^*_\A(\x)$. For~$D_\A^*(\x)$ and~$D'_\A$ this holds since~$D^*_\A(\x)\subseteq D'_\A(\x)$. The same is true for~$D'_\A(\x)$ and~$D_\A(\x)$.
\end{proof}

\begin{thm}\label{theorem:maxnpkernel}
Let~$\Q\in\MNP$. The standard parameterization~$\pQ$ of~$\Q$ admits a polynomial kernelization.
\end{thm}

\begin{proof}
The proof is organized in three parts. First, given an instance~$(\A,k)$ of~$\pQ$, we construct an instance~$(\A',k)$ of~$\pQ$ in time polynomial in the size of~$(\A,k)$. In the second part, we prove that~$(\A,k)$ and~$(\A',k)$ are equivalent. In the third part, we conclude the proof by showing that the size of~$(\A',k)$ is bounded by a polynomial in~$k$.

\noindent(I.) Let~$(\A,k)$ be an instance of~$\pQ$. We use the sets~$D_\A(\x)$ and~$D^*_\A(\x)$ according to Definition~$\ref{def:sod}$ and Proposition~$\ref{prop:dstar}$. Recall that~$D_\A(\x)$ is the set of all disjuncts of~$\psi_{\x,\y}(\S)$ for~$\y\in Y_\A(\x)$. Thus, for each disjunct~$d\in D^*_\A(\x)\subseteq D_\A(\x)$, we can select a~$\y_d\in Y_\A(\x)$ such that~$d$ is a disjunct of~$\psi_{\x,\y_d}(\S)$. Let~$Y'_\A(\x)\subseteq Y_\A(\x)$ be the set of these selected tuples~$\y_d$. Let~$D'_\A(\x)$ be the set of all disjuncts of~$\psi_{\x,\y}(\S)$ for~$\y\in Y'_\A(\x)$. Since~$D^*_\A(\x)$ contains some disjuncts of~$\psi_{\x,\y}(\S)$ for~$\y\in Y'_\A(\x)$ and~$D_\A(\x)$ contains all disjuncts of~$\psi_{\x,\y}(\S)$ for~$\y\in Y_\A(\x)\supseteq Y'_\A(\x)$, we have that~$D^*_\A(\x)\subseteq D'_\A(\x)\subseteq D_\A(\x)$.

For each~$\x$ this takes time~$O(|D^*_\A(\x)|\cdot|Y^*_\A(\x)|)\subseteq O(k^s\cdot|A|^{c_y})$. Computing~$Y'_\A(\x)$ for all~$\x\in A^{c_x}$ takes time~$O(|A|^{c_x}\cdot k^s\cdot|A|^{c_y})$, i.e.\ time polynomial in the size of~$(\A,k)$ since~$k$ is never larger than~$|A|^{c_x}$.\footnote[4]{That is,~$(\A,k)$ is a no-instance if~$k>|A|^{c_x}$ since~$k$ exceeds the number of tuples~$\x\in A^{c_x}$.}

Let~$A'\subseteq A$ be the set of all components of~$\x\in X_\A$ and~$\y\in Y'_\A(\x)$ for all~$\x\in X_\A$. This ensures that~$X_\A\subseteq(A')^{c_x}$ and~$Y'_\A(\x)\subseteq(A')^{c_y}$ for all~$\x\in X_\A$. Let~$R'_i$ be the restriction of~$R_i$ to~$A'$ and let~$\A'=(A',R'_1,\dots,R'_t)$.

\noindent(II.) We will now prove that~$\opt_\Q(\A)\geq k$ if and only if~$\opt_\Q(\A')\geq k$, i.e.\ that~$(\A,k)$ and~$(\A',k)$ are equivalent. Assume that~$\opt_\Q(\A)\geq k$ and let~$\S=(S_1,\dots,S_u)$ such that~$|\lbrace\x:(\A,\S)\models(\exists\y):\psi(\x,\y,\S)\rbrace|\geq k$. This implies that there must exist tuples~$\x_1,\dots,\x_k\in A^{c_x}$ and~$\y_1,\dots,\y_k\in A^{c_y}$ such that~$\S$ satisfies~$\psi_{\x_i,\y_i}(\S)$ for~$i=1,\dots,k$. Thus~$\S$ must satisfy at least one disjunct in each~$\psi_{\x_i,\y_i}(\S)$ since these formulas are in disjunctive normal form. Accordingly let~$d_1,\dots,d_k$ be disjuncts such that~$\S$ satisfies the disjunct~$d_i$ in~$\psi_{\x_i,\y_i}(\S)$ for~$i=1,\dots,k$. We show that there exists~$\S'$ such that:
$$|\lbrace\x:(\A',\S')\models(\exists\y):\psi(\x,\y,\S')\rbrace|\geq k.$$

For~$p=1,\dots,k$ we apply the following step: If~$\y_p\in Y'_\A(\x_p)$ then do nothing. Otherwise consider the partial assignment~$L$ consisting of the at most~$sk$ literals of the disjuncts~$d_1,\dots,d_k$. The set~$D_\A(\x_p)$ contains a disjunct that is satisfiable under~$L$, namely~$d_p$. By Lemma~$\ref{lemma:sandwich}$, it follows that~$D'_\A(\x_p)$ also contains a disjunct satisfiable under~$L$, say~$d'_p$. Let~$\y'_p\in Y'_\A(\x_p)$ such that~$d'_p$ is a disjunct of~$\psi_{\x_p,\y'_p}(\S)$. Such a~$\y'_p$ can be found by selection of~$D'_\A(\x_p)$. Change~$\S$ in the following way to satisfy the disjunct~$d'_p$. For each literal of~$d'_p$ of the form~$S_i(\z)$ add~$\z$ to the relation~$S_i$. Similarly for each literal of the form~$\neg S_i(\z)$ remove~$\z$ from~$S_i$. This does not change the fact that~$\S$ satisfies the disjunct~$d_i$ in~$\psi_{\x_i,\y_i}(\S)$ for~$i=1,\dots,k$ since, by selection,~$d'_p$ is satisfiable under~$L$. Then we replace~$\y_p$ by~$\y'_p$ and~$d_p$ by~$d'_p$. Thus we maintain that~$\S$ satisfies~$d_i$ in~$\psi_{\x_i,\y_i}(\S)$ for~$i=1,\dots,k$.

After these steps we obtain~$\S$ as well as tuples~$\x_1,\dots,\x_k$,~$\y_1,\dots,\y_k$ with~$\y_i\in Y'_\A(\x_i)$, and disjuncts~$d_1,\dots,d_k$ such that~$\S$ satisfies~$d_i$ in~$\psi_{\x_i,\y_i}(\S)$ for~$i=1,\dots,k$. Let~$\S'$ be the restriction of~$\S$ to~$A'$. Then we have that~$(\A',\S')\models\psi_{\x_i,\y_i}(\S')$ for~$i=1,\dots,k$ since~$A'$ is defined to contain the components of tuples~$\x\in X_\A$ and of all tuples~$\y\in Y'_\A(\x)$ for~$\x\in X_\A$. Hence~$\x_i\in\lbrace\x:(\A',\S')\models(\exists\y):\psi(\x,\y,\S')\rbrace$ for~$i=1,\dots,k$. Thus~$\opt_\Q(\A')\geq k$.

For the reverse direction assume that~$\opt_\Q(\A')\geq k$. Since~$A'\subseteq A$ it follows that
$$\lbrace\x:(\A',\S')\models(\exists\y):\psi(\x,\y,\S')\rbrace\subseteq\lbrace\x:(\A,\S')\models(\exists\y):\psi(\x,\y,\S')\rbrace.$$
Thus~$|\lbrace\x:(\A,\S')\models(\exists\y):\psi(\x,\y,\S')\rbrace|\geq k$, implying that~$\opt_\Q(\A)\geq k$. Therefore~$\opt_\Q(\A)\geq k$ if and only if~$\opt_\Q(\A')\geq k$. Hence~$(\A,k)$ and~$(\A',k)$ are equivalent instances of~$\pQ$.

\noindent(III.) We conclude the proof by providing an upper bound on the size of~$(\A',k)$ that is polynomial in~$k$. For the sets~$Y'_\A(\x)$ we selected one tuple~$\y$ for each disjunct in~$D^*_\A(\x)$. Thus~$|Y'_\A(\x)|\leq|D^*(\x)|\in O(k^s)$ for all~$\x\in X_\A$. The set~$A'$ contains the components of tuples~$\x\in X_\A$ and of all tuples~$\y\in Y'_\A(\x)$ for~$\x\in X_\A$. Thus
$$\begin{array}{rcl}|A'|&\leq& c_x\cdot|X_\A|+c_y\cdot\sum_{\x\in X_\A}{|Y'_\A(\x)|}\\
&\leq&c_x\cdot|X_\A|+c_y\cdot|X_\A|\cdot O(k^s)\\
&<&c_x\cdot k\cdot2^s+c_y\cdot k\cdot2^s\cdot O(k^s)=O(k^{s+1}).\end{array}$$

For each relation~$R'_i$ we have~$|R'_i|\leq|A'|^{r_i}\in O(k^{(s+1)r_i})$. Thus the size of~$(\A',k)$ is bounded by~$O(k^{(s+1)m})$, where~$m$ is the largest arity of a relation~$R_i$.
\end{proof}

\begin{remark}
For~$\MSNP$ one can prove a stronger result that essentially relies on Lemma~$\ref{lemma:claim1}$. That way one obtains bounds for the sizes of~$A'$ and~$(\A',k)$ of~$O(k)$ and~$O(k^m)$ respectively.
\end{remark}

\section{Conclusion}\label{conclusion}

We have constructively established that the standard parameterizations of problems in~$\MINF_1$ and~$\MNP$ admit polynomial kernelizations. Thus a strong relation between constant-factor approximability and polynomial kernelizability has been showed for two large classes of problems. It remains an open problem to give a more general result that covers all known examples (e.g.\ \textsc{Feedback Vertex Set}). It might be profitable to consider closures of~$\MSNP$ under reductions that preserve constant-factor approximability. Khanna et al.~\cite{Khanna1998} proved that~APX and~APX-PB are the closures of~$\MSNP$ under~PTAS-preserving reductions and~E-reductions, respectively. Since both classes contain \textsc{Bin Packing} which does not admit a polynomial kernelization, this leads to the question whether polynomial kernelizability or fixed-parameter tractability are maintained under restricted versions of these reductions.

\section*{Acknowledgement} We would like to thank Iyad Kanj for pointing out \textsc{Bin Packing} as an example that is constant-factor approximable but does not admit a polynomial kernelization.

\bibliographystyle{plain}
\bibliography{references}

\begin{thebibliography}{10}

\bibitem{Abu-Khzam2007}
Faisal~N. Abu-Khzam.
\newblock Kernelization algorithms for d-hitting set problems.
\newblock In Frank K. H.~A. Dehne, J{\"o}rg-R{\"u}diger Sack, and Norbert Zeh,
  editors, {\em WADS}, volume 4619 of {\em LNCS}, pages 434--445. Springer,
  2007.

\bibitem{Arora1998}
Sanjeev Arora, Carsten Lund, Rajeev Motwani, Madhu Sudan, and Mario Szegedy.
\newblock Proof verification and the hardness of approximation problems.
\newblock {\em J. ACM}, 45(3):501--555, 1998.

\bibitem{Bafna1999}
Vineet Bafna, Piotr Berman, and Toshihiro Fujito.
\newblock A 2-approximation algorithm for the undirected feedback vertex set
  problem.
\newblock {\em SIAM J. Discrete Math.}, 12(3):289--297, 1999.

\bibitem{Bodlaender2007}
Hans~L. Bodlaender.
\newblock A cubic kernel for feedback vertex set.
\newblock In Wolfgang Thomas and Pascal Weil, editors, {\em STACS}, volume 4393
  of {\em LNCS}, pages 320--331. Springer, 2007.

\bibitem{Bodlaender2008}
Hans~L. Bodlaender, Rodney~G. Downey, Michael~R. Fellows, and Danny Hermelin.
\newblock On problems without polynomial kernels (extended abstract).
\newblock In Luca Aceto, Ivan Damg{\aa}rd, Leslie~Ann Goldberg, Magn{\'u}s~M.
  Halld{\'o}rsson, Anna Ing{\'o}lfsd{\'o}ttir, and Igor Walukiewicz, editors,
  {\em ICALP (1)}, volume 5125 of {\em LNCS}, pages 563--574. Springer, 2008.

\bibitem{Cai1997}
Liming Cai and Jianer Chen.
\newblock On fixed-parameter tractability and approximability of {NP}
  optimization problems.
\newblock {\em J. Comput. Syst. Sci.}, 54(3):465--474, 1997.

\bibitem{Cai06}
Liming Cai and Xiuzhen Huang.
\newblock Fixed-parameter approximation: Conceptual framework and
  approximability results.
\newblock In Hans~L. Bodlaender and Michael~A. Langston, editors, {\em IWPEC},
  volume 4169 of {\em Lecture Notes in Computer Science}, pages 96--108.
  Springer, 2006.

\bibitem{Chen1999}
Jianer Chen, Iyad~A. Kanj, and Weijia Jia.
\newblock Vertex cover: Further observations and further improvements.
\newblock {\em J. Algorithms}, 41(2):280--301, 2001.

\bibitem{DF2006}
Rod~G. Downey and M.~R. Fellows.
\newblock {\em Parameterized Complexity (Monographs in Computer Science)}.
\newblock Springer, November 1998.

\bibitem{Erdos1960}
Paul Erd{\H{o}}s and Richard Rado.
\newblock Intersection theorems for systems of sets.
\newblock {\em J. London Math. Soc.}, 35:85--90, 1960.

\bibitem{Fagin1974}
Ronald Fagin.
\newblock Generalized first-order spectra and polynomial-time recognizable
  sets.
\newblock In {\em Complexity of Computation}, volume~7 of {\em SIAM-AMS
  Proceedings}, pages 43--73, SIAM, Philadelphia, PA, 1974.

\bibitem{Feige2008}
Uriel Feige, MohammadTaghi Hajiaghayi, and James~R. Lee.
\newblock Improved approximation algorithms for minimum weight vertex
  separators.
\newblock {\em SIAM J. Comput.}, 38(2):629--657, 2008.

\bibitem{Flum2006}
J.~Flum and M.~Grohe.
\newblock {\em Parameterized Complexity Theory (Texts in Theoretical Computer
  Science. An EATCS Series)}.
\newblock Springer, March 2006.

\bibitem{Fortnow2008}
Lance Fortnow and Rahul Santhanam.
\newblock Infeasibility of instance compression and succinct {PCPs} for {NP}.
\newblock In Richard~E. Ladner and Cynthia Dwork, editors, {\em STOC}, pages
  133--142. ACM, 2008.

\bibitem{Halperin2000}
Eran Halperin.
\newblock Improved approximation algorithms for the vertex cover problem in
  graphs and hypergraphs.
\newblock {\em SIAM J. Comput.}, 31(5):1608--1623, 2002.

\bibitem{Khanna1998}
Sanjeev Khanna, Rajeev Motwani, Madhu Sudan, and Umesh~V. Vazirani.
\newblock On syntactic versus computational views of approximability.
\newblock {\em SIAM J. Comput.}, 28(1):164--191, 1998.

\bibitem{Kolaitis1994}
Phokion~G. Kolaitis and Madhukar~N. Thakur.
\newblock Logical definability of {NP} optimization problems.
\newblock {\em Inf. Comput.}, 115(2):321--353, 1994.

\bibitem{Kolaitis1995}
Phokion~G. Kolaitis and Madhukar~N. Thakur.
\newblock Approximation properties of {NP} minimization classes.
\newblock {\em J. Comput. Syst. Sci.}, 50(3):391--411, 1995.

\bibitem{Natanzon2000}
Assaf Natanzon, Ron Shamir, and Roded Sharan.
\newblock A polynomial approximation algorithm for the minimum fill-in problem.
\newblock {\em SIAM J. Comput.}, 30(4):1067--1079, 2000.

\bibitem{Niedermeier2006}
Rolf Niedermeier.
\newblock {\em Invitation to Fixed Parameter Algorithms (Oxford Lecture Series
  in Mathematics and Its Applications)}.
\newblock {Oxford University Press, USA}, March 2006.

\bibitem{Nishimura2004}
Naomi Nishimura, Prabhakar Ragde, and Dimitrios~M. Thilikos.
\newblock Smaller kernels for hitting set problems of constant arity.
\newblock In Rodney~G. Downey, Michael~R. Fellows, and Frank K. H.~A. Dehne,
  editors, {\em IWPEC}, volume 3162 of {\em LNCS}, pages 121--126. Springer,
  2004.

\bibitem{Papadimitriou1993}
Christos~H. Papadimitriou.
\newblock {\em Computational Complexity}.
\newblock {Addison Wesley}, November 1993.

\bibitem{Papadimitriou1991}
Christos~H. Papadimitriou and Mihalis Yannakakis.
\newblock Optimization, approximation, and complexity classes.
\newblock {\em J. Comput. Syst. Sci.}, 43(3):425--440, 1991.

\end{thebibliography}

\end{document}